\documentclass[11pt, reqno]{amsart}
\usepackage{amsmath, amsthm, a4, latexsym, amssymb}

\def\@rmrk#1#2{\refstepcounter
    {#1}\@ifnextchar[{\@yrmrk{#1}{#2}}{\@xrmrk{#1}{#2}}}

\makeatletter\@addtoreset{equation}{section}\makeatother

 \newfont{\bfit}{cmbxti10 scaled 2000}
 \newfont{\biggi}{cmr12 scaled 2000}

 \newcommand{\eps}{\varepsilon}

 \newcommand{\R}{\mathbb{R}}
 
 \newcommand{\N}{\mathbb{N}}

 \newcommand{\prob}{\mathbb{P}}

 \renewcommand{\P}{\mathbb{P}}
 
 \newcommand{\skria}{{\mathcal A}}
 \newcommand{\skrib}{{\mathcal B}}
 \newcommand{\skric}{{\mathcal C}}
 \newcommand{\skrid}{{\mathcal D}}
 
 \newcommand{\skrie}{{\mathcal E}}
 \newcommand{\skrif}{{\mathcal F}}
 
 \newcommand{\skrih}{{\mathcal H}}

 \newcommand{\skrik}{{\mathcal K}}
 
 \newcommand{\skrim}{{\mathcal M}}
 
 \newcommand{\skrit}{{\mathcal T}}

 \newcommand{\skriv}{{\mathcal V}}
 \newcommand{\skrix}{{\mathcal X}}
 \newcommand{\skriy}{{\mathcal Y}}

 \newcommand{\sfrac}[2]{\mbox{$\frac{#1}{#2}$}}

\def\1{{\mathchoice {1\mskip-4mu\mathrm l}      
{1\mskip-4mu\mathrm l}
{1\mskip-4.5mu\mathrm l} {1\mskip-5mu\mathrm l}}}

\newcommand{\eq}{\begin{equation}}
\newcommand{\en}{\end{equation}}

{\nopagebreak {\hspace*{\fill}\rule{2mm}{2mm}}\\ }

{\nopagebreak {\hspace*{\fill}\rule{2mm}{2mm}}\\ }

\renewcommand{\subsection}{\secdef \subsct\sbsect}
\newcommand{\subsct}[2][default]{\refstepcounter{subsection}
\vspace{0.15cm}
{\flushleft\bf \arabic{section}.\arabic{subsection}~\bf #1  }
\nopagebreak\nopagebreak}
\newcommand{\sbsect}[1]{\vspace{0.1cm}\noindent
{\bf #1}\vspace{0.1cm}}

\newtheorem{theorem}{Theorem}[section]
\newtheorem{lemma}[theorem]{Lemma}
\newtheorem{cor}[theorem]{Corollary}

\newtheoremstyle{thm}{1.5ex}{1.5ex}{\itshape\rmfamily}{}
{\bfseries\rmfamily}{}{2ex}{}

\newtheoremstyle{rem}{1.3ex}{1.3ex}{\rmfamily}{}
{\itshape\rmfamily}{}{1.5ex}{}
\theoremstyle{rem}
\newtheorem{remark}{{\slshape\sffamily Remark}}[]

\refstepcounter{subsection}

\def\thebibliography#1{\section*{References}
  \list%
  {\arabic{enumi}.}
    {\settowidth\labelwidth{[#1]}\leftmargin\labelwidth
    \advance\leftmargin\labelsep
    \parsep0pt\itemsep0pt
    \usecounter{enumi}}
    \def\newblock{\hskip .11em plus .33em minus .07em}
    \sloppy                   
    \sfcode`\.=1000\relax}

\begin{document}
\title[LLLD  for multitype Galton-Watson  process]
{\Large Local Large  deviations: MCMillian Theorem for  multitype Galton-Watson  Processes}

\author[Kwabena Doku-Amponsah]{}

\maketitle
\thispagestyle{empty}
\vspace{-0.5cm}

\centerline{\sc{By Kwabena Doku-Amponsah}}
\renewcommand{\thefootnote}{}
\renewcommand{\thefootnote}{1}
\renewcommand{\thefootnote}{}
\footnote{\textit{Address:} Statistics Department, University of
Ghana, Box LG 115, Legon,Ghana.\,
\textit{E-mail:\,kdoku@ug.edu.gh}.}
\renewcommand{\thefootnote}{1}
\centerline{\textit{University of Ghana}}

\begin{quote}{\small }{\bf Abstract.}
In  this  article  we prove  a  local  large deviation  principle (LLDP)  for  the critical multitype  Galton-Watson  process  from spectral  potential  point. We  define  the  so-called  a spectral  potential $U_{\skrik}(\,\cdot,\,\pi)$ for   the Galton-Watson  process, where $\pi$   is  the normalized  eigen  vector  corresponding  to  the  leading \emph{Perron-Frobenius eigen  value } $\1$  of  the  transition matrix $\skria(\cdot,\,\cdot)$  defined from ${\skrik},$ the  transition  kernel. We show  that  the Kullback  action or  the  deviation function, $J(\pi,\rho),$ with  respect  to an  empirical offspring measure,  $\rho,$   is  the Legendre  dual  of   $U_{\skrik}(\,\cdot,\,\pi).$  From  the LLDP  we  deduce  a conditional  large  deviation  principle  and  a  weak  variant  of  the  classical  McMillian Theorem  for  the  multitype   Galton-Watson  process. To  be  specific,  given  any  empirical  offspring  measure  $\varpi,$   we  show  that  the  number  of  critical multitype  Galton-Watson  processes   on  $n$  vertices is  approximately  $e^{n\langle \skrih_{\varpi},\,\pi\rangle},$   where  $\skrih_{\varpi}$  is a suitably  defined  entropy.

\end{quote}\vspace{0.5cm}

\section{Background}

\subsection{Introduction}

We revisit the typed trees  model  described
by the following procedure: The root carries a random type chosen
according to some law on a finite alphabet; given the type of a
vertex, the number and types of the children (ordered from left to
right) are given, independently of everything else, by an offspring
law.  We  shall  refer  to  the  tree  together  with  the  types  on  the  tree  as a  multitype Galton-Watson  Process.\\

Presently   there  exists  some  Large  deviation  Principle(LDP)  and  Basic Information  Theory  for  the  multitype Galton-Watson  Process.See, example   Dembo et al.~\cite{DMS03},  Doku-Amponsah~\cite{DA17b},  Doku-Amponsah~\cite{DA12},  Doku-Amponsah~\cite{DA17a}. In  \cite{DMS03}, LDPs  were  proved for  empirical   measures  of  the  multitype Galton-Watson trees with the  exponential  moments of  offspring transition kernel being finite, in a topology  stronger  than  weak topology. See \cite[pp.~6]{DMS03}.  \cite{DA12}  proved  an  Asymptotic  Equipatition  Property  for  hierarchical  Data  Structures,  using  \cite[Theorem~2.2]{DMS03},  for  bounded  offspring  transition  kernels. Recently, Doku-Amponsah~\cite{DA17a} found  a  Lossy version of  the result \cite[Theorem~2.1]{DA17a}. The  article  \cite{DA17b}  extended  the LDP  result of  \cite[Theorem~2.2]{DMS03} for the multitype Galton-Watson Process  to  cover  offspring  laws (e.g. geometric \sfrac{1}{2}), which  were  not  covered by  \cite{DMS03}, in  the  weak  topology.\\

  In  this  article   we prove  that the   LLDP for  the  multitype Galton-Watson process  conditioned  on  the  empirical  offspring  measure  of  the   process. See  Bakhtin~\cite[Theorem~3.1]{BIV15}  for  similar  results  for  the  empirical  measure  of  iid random  variables. The  current  article,  like  the  paper by Baktain \cite{BIV15},  differs  from most of the  aformentioned   articles,  by the  lack  of  any topology  on  the  range  of  random  variables and  by the use  of  the  weak  topology  on  the  set  of  measures  generated  by bounded  measurable  functions.  The  main  technique  used  in  this  article  is  spectral  potential  theory. See \cite{BIV15}.  To  be  specific  about  this  approach,  we  define  a  spectral  potential  of  the  multitype Galton-Watson  trees  and  use  it to define  an  extension  of  the  relative  entropy,  which  we  also  call  the  Kullback  action. The  Kullback  action  has  the  rate  function  of  the  LDPs  in  \cite{DA17b}  and  \cite{DMS03} as  special  cases. By  a  proper  exponential  change  of  measure,using  the  Kullback  action  we  prove  the  LLDP, and  from  the  LLDP  we  obtain  the  classical  McMillan-Breiman  Theorem  as  a particular  case.  The  conditional  LDP  is  also  derived  from  the LLDP.

\subsection{The  Multitype  Galton-Watson  Process.}\label{SHS}
 Let us
adopt some concepts from the paper Dembo et al.
\cite{DMS03}. We shall  denote  by  $GW$  the set of all finite rooted
planar trees $\skrit$, by $\skriv=\skriv(\skrit)$ the set of all vertices and by
$\skrie=\skrie(\skrit)$ the set of all edges oriented away from the root, which is
always denoted by $\eta$. We write $|\skriv(\skrit)|$ for the number of vertices
in the tree $\skrit.$ Let $\skriy$ be a finite alphabet with a chosen $\sigma-$ field  $\skrif$ of  subsets and write
$$\displaystyle\skriy^*=\bigcup_{n=0}^\infty \{n\} \times
\skriy^n.$$ 
Note that the
child  of any vertex $v\in \skriv(\skrit)$ is characterized by an element of
$\skriy^*$ and that there is an element $(0,\emptyset)$ in
$\skriy^*$ symbolizing the absence of a  child.

Let $\beta$ be a probability distribution on
$\skriy$ and ${\skrik}:\skriy\times\skriy^{*}\rightarrow[0,\,1]$
be an offspring transition kernel.  We describe  the law $\prob$ of a multitype Galton-Watson process $Y,$  see Mode~\cite{Mo71}, by the following procedure:
\begin{itemize}
\item The root $\eta$ carries a random type $Y(\eta)$ chosen according
to the probability measure $\beta$ on $\skriy.$
\item For every vertex with type $a\in\skriy$ the offspring number
and types are given independently of everything else, by the
offspring law ${\skrik}\{\cdot\,|\,a\}$  on $\skriy^{*}.$ We
write
$$
{\skrik}\big\{\cdot\,|\,a\}={\skrik}\big\{(N,Y_1,...,Y_N)\in\,\cdot\,|\,a\big\},$$
i.e. we have a random number $N$ of descendants  with types
$Y_1,...,Y_N.$
\end{itemize}
We shall consider $Y=((Y(v),\,C(v)),\,v\in \skriv)$  under the joint law
of tree and offspring. Interpret $Y$ as a \emph{multitype
Galton-Watson process } and $Y(v)$ as the type of vertex $v.$ For each
typed tree $Y$ and each vertex $v$ we denote by $\displaystyle
C(v)=(N(v),Y_1(v),\ldots,Y_{N(v)}(v))\in \skriy^*,$ the number and
types of the offsprings of $v$, ordered from left to right. Notice
that the children of the root (denoted by $\eta$) are ordered, but
the root itself is not. 
Denote, for every $c=(n(c),a_1(c),\ldots,a_n(c))\in\skriy^{*}$ and $
a\in\skriy,$ the \emph{multiplicity} of the symbol $a$ in $c$ by
 $$m(a,c)=\sum_{i=1}^{n(c)}1_{\{a_{i}=a\}}.$$

Define the matrix $\skria$ with index set $\skriy\times\skriy$ and
nonnegative entries by
$$\skria(a,b)=\sum_{c\in\skriy^{*}}{\skrik}\{c\,|\,b\}m(a,c),\mbox{
for $a,b\in\skriy.$}$$

$\skria(a,b)$ is the expected number of offspring of type $a$ of a vertex
of type $b.$  Let
$\skria^{*}(a,b)=\sum_{k=1}^{\infty}\skria^{k}(a,b)\in[0,\infty].$ We say that
the matrix $\skria$ is irreducible if $\skria^{*}(a,b)>0;$ for all
$a,b\in\skriy.$

The multitype Galton-Watson Process is called irreducible if the matrix
$\skria$ is irreducible. It is called critical $($subcritical,
supercritical$)$ if the largest eigenvalue of the matrix $A$ is $1 $
$($ less than $1,$ greater than $1$ resp.$).$ Let $\pi$ be the
\emph{eigenvector} corresponding to the largest  \emph{Perron-Frobenius eigenvalue
} $1$ (normalized to a probability vector).  Then, by the Perron-Frobenius Theorem, $\pi$ is
\emph{unique}, if the Galton-Watson tree is irreducible.  See  Dembo et al.~\cite[Theorem~3.1.1]{DZ98}.

For every multitype Galton-Watson tree $Y,$ the \emph{empirical
offspring measure }$M_Y$ is defined by
$$M_Y(a,c)=\frac{1}{|T|}\sum_{v\in V}\delta_{(Y(v),C(v))}(a,c),\,\mbox{ for
$(a,c)\in\skriy\times\skriy^{*}$}.$$ We call $\rho$
\emph{shift-invariant} if\,
$\displaystyle\rho_{1}(a)=\sum_{(b,c)\in\skriy\times\skriy^*}m(a,c)\rho(b,c),\,\mbox{
for all $a\in\skriy\times\skriy^{*}$ }.$

 We denote by $\skrim(\skriy\times\skriy^*)$ the space of
probability measures $\rho$ on $\skriy\times\skriy^*$ with $\int n \,
\rho(da \, , dc)<\infty$, using the convention
$c=(n,a_1,\ldots,a_n)$. Denote  by  $\skrim_*(\skriy\times\skriy^*)$  the  space  of   positive  measures  on   $\rho$ on $\skriy\times\skriy^*$ with $\int n \,\rho(da \, , dc)<\infty,$ by $\skrib(\skriy\times\skriy^*)$ the  space  of  all  real-valued bounded  measurable functions  on  $\skriy\times\skriy^*,$    by $\skrib_*(\skriy\times\skriy^*)$ the  space  of  continuous linear functionals on  $\skrib(\skriy\times\skriy^*)$ and  by $\skrib_{+}(\skriy\times\skriy^*)$ the  collection  of  all  positive linear  functionals  on  $\skrib(\skriy\times\skriy^*).$


The  remaining  part  of  the  article  is  organized  in  the  following  manner: Section~\ref{AEP}  contains  the  main  results  of  the  article:  Theorem~\ref{smb.tree}, Corollary~\ref{smb.tree1}  and Theorem~\ref{smb.tree2}. In  Section~\ref{Proofmain} these results  are  proved.

\section{Statement of main results}\label{AEP}
 Write  $\displaystyle \langle f\,,\,\sigma\rangle:=\sum_{y\in\skriy}\sigma(y)f(y)$  and   define  the  \emph{spectral  potential}  $U_{\skrik}(g,\,\pi)$  of  the multitype  Galton-Watson process  $Y$ by
\begin{equation}\label{Sp}
U_{\skrik}(g,\,\pi)=\log\langle e^{g},\,\pi\otimes{\skrik}\rangle.
\end{equation}

Observe  that  \ref{Sp}  possesses  all  the  remarkable  properties  mentioned  in  \cite[pp.536-538]{BIV15}.i.e.  (i) It  is  finite  on  $\displaystyle \skrid(\pi):=\Big\{ g:\skriy\times\skriy_*\to \R\, | \langle e^{g},\,\pi\otimes{\skrik}\rangle<\infty\Big\}$  (ii)  It  is  monotone (iii)  it  is  additively  homogeneous  and  it  is  convex  in $g.$   For $\rho\in\skrib(\skriy\times\skriy_*)$ we call
a nonlinear  functional

\begin{align}\label{Kullback}
J(\pi,\,\rho):=\left\{\begin{array}{ll}H(\rho\,\|\,\pi\otimes{\skrik})
& \mbox{ if  $\rho$ is shift-invariant and $\rho_1=\pi$,}\\
+\infty & \mbox{otherwise.}
\end{array}\right.
\end{align}

\emph{Kullback  action.}

 We denote by $ P_n(y)=\prob_n\{Y=y\}=\prob\Big\{Y=y\,\big |\,|\,\skriv|=n\Big\}$ the
distribution of the multitype Galton-Watson tree  $x$ conditioned to
have $n$
vertices.  In  Theorem~\ref{smb.tree}  below  we  state  our  main  result,  the  LLDP  for  the multitype Galton-Watson tree.
\begin{theorem}[LLDP]\label{smb.tree}
Suppose $y=(y(v):v\in \skriv)$ is an irreducible, critical multitype
Galton-Watson tree on  $n$ vertices, with finite  type space $\skriy$ and an
offspring kernel ${\skrik}$. Then,
\begin{itemize}

\item[(i)] for any  functional  $\rho\in \skrim_*(\skriy\times\skriy_*)$  and a  number  $\eps>0,$  there  exists  a  weak  neighborhood  $B_{\rho}$  such  that
$$ P_n\Big\{y\in GW\,\Big |\, M_y\in B_{\rho}\Big\}\le e^{-nJ(\pi,\,\rho)-n\eps+o(n)}.$$
\item[(ii)] for  any  $\rho\in\skrim_*(\skriy\times\skriy_*)$, a  number  $\eps>0$  and  a fine  neighborhood  $B_{\rho}$  we  have  the asymptotic  estimate:
    $$ P_n\Big\{y\in GW\,\Big |\, M_y\in B_{\rho}\Big\}\ge e^{-nJ(\pi,\,\rho)+n\eps-o(n)}.$$
    \end{itemize}
\end{theorem}

Next  we  state  the  McMillian-Breiman Theorem  for  the multitype Galton-Watson tree. 
For  any  $a\in\skriy$   and  the  conditional probability  measure $\varpi(\cdot\,|\,a)$   on  $\skriy_0^*$ we  define  an  entropy   by
$$\skrih_\varpi(a):= -\sum_{c\in\skriy_0^*}\varpi(c\,|\,a)\log\, \varpi(c\,|\,a).$$  We write
$$\displaystyle\skriy_0^*=\bigcup_{n=0}^{N_0}\{n\} \times
\skriy^n$$  and  state a corollary  of  Theorem~\ref{smb.tree} below:
\begin{cor}[McMillian Theorem]\label{smb.tree1} Suppose  $\skriy_0$  is  the  space  of  all  irreducible, critical multitype
Galton-Watson trees with finite  type space $\skriy$ and an
offspring kernel supported  on  $\skriy_0^*.$
\begin{itemize}

\item[(i)] For  any empirical offspring  measure  $\varpi$   on  $\skriy\times\skriy_0^*$   and  $\eps>0,$  there  exists  a neighborhood  $B_{\varpi}$  such  that
$$ Card\Big(\big\{y\in\skriy_0\,| \,M_y\in B_{\varpi}\big\}\Big)\ge e^{n\big(\langle \skrih_\varpi\,,\,\pi\rangle+\eps\big)}.$$
\item[(ii)] for any  neighborhood  $B_{\varpi}$   and  $\eps>0,$  we  have
   $$Card\Big(\big\{y\in\skriy_0\,|\, M_y\in B_{\varpi}\big\}\Big)\le e^{n\langle \skrih_\varpi\,,\,\pi\rangle-\eps\big)},$$
    \end{itemize}
\end{cor}

where $Card(A)$  means  the  cardinality  of  $A.$

\begin{remark}
Note  that  if  the  transition  kernel  of  the  multitype  Galton-Watson process  is  bounded. Thus, there  exists $N_0\in\N$  such  that
${\skrik}\Big\{N>N_0\,|\,a\Big\}=0, \,\mbox{ for  all   $a\in\skriy$},$
then, we  could  recover  the  results  of  \cite[Theorem~2.1]{DA12}  by  setting  $\varpi(\cdot\,|\,a)={\skrik}\big\{\cdot\,|\,a\big\},$ for $a\in\skrix.$ Thus,  we  have

$$ Card\Big(\Big\{y\in\skriy_0\,\Big\}\Big)\approx e^{n\langle \skrih_{\skrik}\,,\,\pi\rangle}.$$

\end{remark}

Finally,  we  state  in  Theorem~\ref{smb.tree2}  the  full  LDP  for  the multitype Galton-Watson tree.
\begin{theorem}[LDP]\label{smb.tree2}
Suppose $y=(y(v):v\in \skriv)$ is an irreducible, critical multitype
Galton-Watson tree with finite  type space $\skriy$ and
offspring kernel ${\skrik}$.

\begin{itemize}

\item[(i)]  Let  $F$ be an open subset  of  $ \skrim(\skriy\times\skriy^*)$.  Then  we  have
$$\lim_{n\to\infty}\frac{1}{n}\log P_n\Big\{y\in GM\,\Big |\,M_y\in F\Big\}\ge - \inf_{\rho\in F}J(\pi,\,\rho).$$

\item[(ii)] Let  $\Gamma$ be  a closed subset  of  $ \skrim(\skriy\times\skriy^*)$. The  we  have

    $$ \lim_{n\to\infty}\frac{1}{n}\log P_n\Big\{y\in GM\,\Big |\, M_y\in \Gamma\Big\}\le-\inf_{\rho\in \Gamma}J(\pi,\,\rho).$$

    \end{itemize}

\end{theorem}
\section{Proof  Main of  Results}\label{Proofmain}

\subsection{Properties  of  the  Kullback  action.}
Next,  we  state  a key  ingredient (Lemma~\ref{Lem1}) in the  proof  of our  main result, Theorem~\ref{smb.tree}.  This  Lemma  gives  remarkable   properties  of  \ref{Kullback} above,  which  will  help  us  circumvent  the  topological  problems  faced  in \cite{DA17b}  and  \cite{DMS03}. To  state  the  lemma,  we denote  by  $\skric$   the  space  of  continuous  functions    $g:\skriy\times\skriy_*\to \R$  and notice that, the  proof  below  follows  similar  ideas as  the  proof  of  \cite[Lemma~2.2]{BIV15}  for  the  empirical  measures on  measurable spaces.

\begin{lemma}\label{Lem1}\label{LLDP1}The  following  holds  for  the  Kullback action or divergence  function  $J(\pi,\,\rho).$
\begin{itemize}
\item [(i)]$\displaystyle J(\pi,\,\rho)=\sup_{g\in\skric}\Big\{\langle g,\,\rho\rangle-U_{\skrik}(g,\,\pi) \Big\}.$
\item[(ii)] The  function $J(\pi,\,\cdot)$  is  convex and  lower  semi-continuous on  the  space  $\skrib_*(\skriy\times\skriy_*).$
\item[(iii)] For  any  real  $c,$  the  set  $\Big\{\rho\in\skrib_*(\skriy\times\skriy_*):\, J(\pi,\,\rho)\le  c\Big\}$  is  weakly  compact.
\end{itemize}
\end{lemma}

\begin{proof}

(i)   Let   $g\in\skric$  be  such  that  $\langle g,\,\rho\rangle$  approximates  the  functional $\langle \phi,\,\rho\rangle$
and  $U_{\skrik}(g,\,\pi)$  approximates $U_{\skrik}(\phi,\,\pi)$  where  $\phi\in\skrib(\skriy\times\skriy_*).$   Suppose  $\rho$  is  absolutely  continuous  with  respect  to  $\pi\otimes{\skrik}.$  Thus,  there  exists  $g$  such  that  $\rho=g\pi\otimes{\skrik}.$   For $t>0,$  we define  the  approximating function   $g_t\in\skrib(\skriy\times\skriy^*)$
as  follows
 \begin{equation}
 \begin{aligned}\label{funct1}
g_t(a,c):=\left\{\begin{array}{ll}\log\,g(a,c),
& \mbox{ if  $e^{-t}<g(a,c)<e^{t}$,}\\
t, &\mbox{ if $g(a,c)>e^t$} \\
-t, &\mbox{ if $g(a,c)<e^{-t}$}
\end{array}\right.
\end{aligned}
\end{equation}
for  all  $(a,c)\in\skriy\times\skriy^*.$
 Now,  for  $t\to\infty$   we  have
 $$\begin{aligned}
 \langle e^{g_t},\,\pi\otimes\skrik\rangle=\int g\1_{\{e^{-t}<g<e^{t}\}}\pi\otimes\skrik(da,dc)&+\int e^{t} \1_{\{g>e^{t}\}}\pi\otimes\skrik(da,dc)\\
 &+\int e^{-t}\1_{\{e^{-t}>g\}}\pi\otimes\skrik(da,dc)\to \langle g,\,\pi\otimes\skrik\rangle=\langle \1,\,\rho\rangle=1
 \end{aligned}$$

$$\begin{aligned}
 \langle g_t,\,\pi\otimes\skrik\rangle &=\int  g\log g\1_{\{e^{-t}<g<e^{t}\}}\pi\otimes\skrik(da,dc)\\
 &+\int t \1_{\{g>e^{t}\}}\pi\otimes\skrik(da,dc)+\int  -t\1_{\{e^{-t}>g\}}\pi\otimes\skrik(da,dc)
 \to \langle g\log g,\,\pi\otimes\skrik\rangle=\langle \log g,\,\rho\rangle.
 \end{aligned}$$
Therefore  we  have  $\lim_{t\to\infty}\big(\langle g_t,\,\pi\otimes\skrik\rangle-\langle e^{g_t},\,\pi\otimes\skrik\rangle\big)\to J(\pi,\,\rho)$  which proves  Lemma~\ref{LLDP1}(i).

 Suppose  $\rho$  is not  absolutely  continuous  with  respect  to  $\pi\otimes{\skrik}.$ Thus,  there  exists  an  $\eps>0$  such  that for  any  $\eta>0$  there  exists  $B_{\delta}\subset \skriy\times\skriy^*$  with  $\pi\otimes\skrik(B_{\delta})\le \delta,$    and at  the  same  time  we  have $\rho(B_{\delta})>\eps.$  For this  $\delta$   we define  the  function

 \begin{equation}
 \begin{aligned}\label{funct1}
g_{\delta}(a,c):=\left\{\begin{array}{ll}-\log\,\delta
& \mbox{ if  $(a,c)\in B_{\delta} $,}\\
0, &\mbox{ if $(a,c)\notin B_{\delta}.$}
\end{array}\right.
\end{aligned}
\end{equation}
Then  we  have
$\displaystyle\lim_{t\to\infty}\big(\langle g_{\delta},\,\pi\otimes\skrik\rangle-\langle e^{g_{\delta}},\,\pi\otimes\skrik\rangle\big)\ge -\eps\log \delta-\langle e^{g_{\delta}},\,\pi\otimes\skrik\rangle\ge  -\eps\log \delta-\log(2).$

Taking  limit  as  $\delta\downarrow 0$  we  have   $J(\pi,\rho)=+\infty,$  which ends  the  proof  of Lemma~\ref{Lem1} (i).\\

(ii)\& (iii). Observe  from  the  variational  formulation  of  the relative  entropy,  see Dembo et al.~ \cite{DZ98},  and  Lemma~\ref{Lem1}(i) that  $J(\pi,\rho)$  reduces  to equation \ref{Kullback} above.  Now  the  relative  entropy  is  convex  and  lower  semi-continuous, and  all its  level  sets  are  compact. Hence  we  have  $J(\pi,\rho)$  convex  and  lower  semi-continuous,  and  all  its  level  sets  are  weakly  compact in  the  weak  topology,   which  ends  the  proof  of the  Lemma.

\end{proof}

Note that  Lemma~~\ref{Lem1} (i) above  implies  the  so-called  variational  principle.See,  example  Komkov~\cite{KV86}.

\subsection{Proof  of  Theorem~\ref{smb.tree}.}
By   Lemma~\ref{LLDP1},  for  any  $\eps>0$  there  exists  a  function  $g\in\skrib(\skriy\times\skriy_*)$  such  that  $$J(\pi,\rho)-\sfrac{\eps}{2} < \langle g,\,\rho\rangle-U_{\skrik}(g,\,\pi).$$

Let us define  the  probability  distribution $\tilde{P}$   by
$$\tilde{P}_n(y)=\tilde{\P}\big\{y, |\skriv|=n\big\}/\tilde{\P}\big\{|\skriv|=n\big\}=\pi(y(\eta))\prod_{v\in V,|\skriv|=n} e^{g(y(v),c(v))-U_{\skrik}(g,\pi)}{\skrik}\Big\{c(v)\Big |\,y(v)\Big\}/\tilde{\P}\big\{|\skriv|=n\big\}.$$  Then,  it  is  not  too  hard  to  observe  that  we  have  $$ \frac{dP_n(y)}{d\tilde{P}_n(y)}
=\prod_{v\in V,|\skriv|=n}e^{-g(y(v),c(v))+U_{\skrik}(g,\pi)}=e^{-\sum_{v\in V,|\skriv|=n}g(y(v),c(v))+nU_{\skrik}(g,\pi) }\sfrac{\tilde{\P}\big\{|\skriv|=n\big\}}{\P\big\{|\skriv|=n\big\}}.$$

Now  we  define a  neighbourhood  of  the  functional  $\rho$  as  follows:
$$B_{\rho}=\Big\{\varpi\in\skrib(\skriy\times\skriy_*):  \langle g, \, \varpi\rangle>\langle g, \, \rho\rangle -\sfrac{\eps}{2}\Big\}.$$

Using  \cite[Lemma~3.1]{DMS03},  under  the  condition  $M_y\in B_{\rho}$    we  have  that  $$ \frac{dP_n(y)}{d\tilde{P}_n(y)}<e^{nU_{\skrik}(g,\,\pi)-n\langle g, \, \rho\rangle +n\sfrac{\eps}{2}+o(n)}<e^{-nJ(\pi,\rho)+n\eps+o(n)}.$$

Hence,  we have

$$ P_n\Big\{y\in\skrit| M_y\in B_{\rho}\Big\}\le \int\1_{\{M_y\in B_{\rho}\}} d\tilde{P}_n(y) \le \int\1_{\{M_y\in B_{\rho}\}} e^{-nJ(\pi,\,\rho)-n\eps+o(n)}d\tilde{P}_n(y) \le e^{-nJ(\pi,\,\rho)-n\eps+o(n)} .$$
Note  that  $ J(\pi,\rho)=\infty$  implies  Theorem~\ref{smb.tree}(ii), so  it  is  sufficient  for  us  to  prove  it  for  a  probability  measure  of  the  form  $\rho=g\pi\otimes{\skrik}$  and  for $J(\pi,\rho)=\langle \log\, g,\,\rho\rangle<\infty$.   Fix  any  number $\eps>0$  and  any  neighbourhood  $B_{\rho}\subset\skrim(\skriy\times\skriy_*).$  We  define  the  sequence  of  sets

$$\skrit_n:=\Big\{y\in\skrit: M_y\in B_{\rho},\,\Big|\int\log g dM_y-\int \log gd\rho\Big|\le  \sfrac{\eps}{2}\Big\}.$$

 Using  \cite[Lemma~3.1]{DMS03}, we observe  that for  all $x\in\skrit_n,$
$$ \frac{dP_n(y)}{d\tilde{P}_n(y)}=\prod_{v\in V}\frac{1}{g(y(v),c(v))}\sfrac{\tilde{\P}\big\{|\skriv|=n\big\}}{\P\big\{|\skriv|=n\big\}}=e^{-n\langle \log g,\,M_y\rangle+o(n)}> e^{-n\langle \log \, g,\,\rho\rangle+n\eps/2+o(n)}$$

This  gives  us $$P_n(\skrit_n)=\int_{\skrit_n}dP_n(x)\ge  \int_{\skrit_n}e^{-n(\langle \log g,\, \rho\rangle +\eps/2)-o(n)}d\tilde{P}_n(x)=e^{-nJ(\pi,\rho)+n\eps-o(n)}\tilde{P}_n(\skrit_n).$$
Using  the  law  of  large  numbers  we  have  $\displaystyle\lim_{n\to\infty}\tilde{P}_n(\skrit_n)=1$  which  completes  the  proof  of  the  Theorem.

\subsection{Proof  of  Corollary~\ref{smb.tree1}.}
The  proof  of  Corollary~\ref{smb.tree1}  follows  from  the  definition  of  the  Kullback action  and  Theorem~\ref{smb.tree} if  we  set   $\rho=\varpi$  and  $\skrik\big\{c\,| a\big\}=1$  for  all  $c=(n,a_1,a_2,a_3,...,a_n)\in\skriy_0^*.$\\

\pagebreak

The  proof  of Theorem~\ref{smb.tree2}   below  follows  from  Theorem~\ref{smb.tree} above  using  similar  arguments  as  in \cite[p. 544]{BIV15}.
\subsection{Proof  of  Theorem~\ref{smb.tree2}.}
\begin{proof}
Note  that  the  empirical  offspring  measure  is  a  probability  measure  and  so belongs  to  the  unit  ball  in  $\skrib_*(\skriy\times\skriy^*).$ Hence,  without  loss  of  generality  we  may  assume  that  the  set  $G$  in  Theorem~\ref{smb.tree2}(ii)  is  relatively  compact. See Lemma~~\ref{Lem1} (iii). Choose  any  $\eps>O.$   Then  for  every  functional  $\rho\in \Gamma$  we can  find  a  weak  neighbourhood  such  that the  estimate  of  Theorem~\ref{smb.tree}(i)  holds.  We  choose  from  all  these  neighbourhood  a  finite  cover  of  $GM$  and  sum up over  the  estimate  in  Theorem~\ref{smb.tree}(i)   to  obtain

$$ \lim_{n\to\infty}\frac{1}{n}\log P_n\Big\{y\in GW\,\Big |\, M_y\in \Gamma\Big\}\le-\inf_{\rho\in \Gamma}J(\pi,\,\rho)+\eps.$$
As $\eps$ was  arbitrarily  chosen  and  the   lower  bound  in  Theorem~\ref{smb.tree}(ii) implies  the  lower  bound  in  Theorem~\ref{smb.tree2}(i), we  have  the  desired  results,  which  ends  the  proof of  the Theorem.

\end{proof}

\section{Conclusion}
In  this  article  we  have  found  an  LLDP  for  the  multitype  Galton-Watson process without  any  topological  restrictions  on  the  space  of probability measures  on $\skriy\times\skriy^*,$  from  a spectral  potential  point. From  the  LLDP   we  deduce other  results  such  as  the  classical  McMillian  Theorem  and  the full  conditional  large  deviation  principle  for  this process. The  main  technique  used   is  exponential  change  of  measure. These results  have  thrown more  insight  on  the  results  of  \cite{DA12} and \cite{DA17b}.\\

{\bf Acknowledgement}

\emph{This  article  was  finalized  at  the  Carnigie Banga-Africa writeshop, June 27-July  2017, in  Koforidua.}




\begin{thebibliography}{WWW98}

\bibitem[1]{BIV15}
{\sc I.V.~Bakhtin.}
\newblock{Spectral Potential,  Kullback  Action,  and  Large  deviations  of  empirical measureson  measureable  spaces.}
\newblock{\emph{Theory  of  Probability  and  application. Vol.  50,No.4.(2015) pp.535-544.}}
\smallskip



\bibitem[2]{DA17a}
{\sc K.~Doku-Amponsah.}
\newblock{Lossy Asymptotic Equipartition Property for hierarchical data structures.}
\newblock {\emph{ Far East  Journal  of  Mathematical  Sciences, Vol. 101, 2017, pp.1013-1024.}}
\smallskip

\bibitem[3]{DA12}
{\sc K.~Doku-Amponsah.}
\newblock{Asymptotic equipartition properties for hierarchical and networked  structures.}
\newblock ESAIM: PS 16 (2012): 114-138.DOI: 10.1051/ps/2010016.
\smallskip
\bibitem[4]{DA17b}
{\sc K.~Doku-Amponsah.}
\newblock{Large deviations for  Multitype Galton-Watson  Trees.}
\newblock { To  appear  \emph{Far  East Journal of  Mathematical  Sciences}}
\smallskip

\bibitem[5]{DMS03}
{\sc A.~Dembo, P.~M\"orters} and {\sc S.~Sheffield.}
\newblock Large deviations of Markov chains indexed by random trees.
\newblock \emph{Ann. Inst. Henri Poincar\'e: Probab.et Stat.41,}
(2005) 971-996.
\smallskip



\bibitem[6]{DZ98}
{\sc A.~Dembo} and {\sc O.~Zeitouni.}
\newblock Large deviations techniques and applications.
\newblock Springer, New York, (1998).
\smallskip

\bibitem[7]{Mo71}
 {\sc C.J.~Mode.}
 \newblock Multitype Branching Processes Theory and Applications.
 \newblock American Elsevier,New York, (1971).
\smallskip

\bibitem[8]{KV86}
{\sc V.~Komkov.}
\newblock{Variational principles of continuum mechanics with engineering applications. Vol. 1. Critical points theory. Mathematics and its Applications, 24.}
\newblock{ D. Reidel Publishing Co., Dordrecht.(1986).}

\end{thebibliography}
\end{document}